\documentclass[journal]{IEEEtran}

\usepackage{cite}
\usepackage{amsmath,amssymb,amsfonts}
\usepackage{algorithmic}
\usepackage{graphicx}
\usepackage{textcomp}

\usepackage{enumitem}
    \usepackage{amssymb}
    \usepackage{amsthm}
    \usepackage{amsfonts}
    \usepackage{braket}
    \usepackage{supertabular}
    \usepackage{amssymb}
    \usepackage{amsmath,amsthm,amssymb}
    \usepackage{lipsum,graphicx,subcaption}

\newtheorem{theorem}{Theorem}

\newtheorem{lemma}[theorem]{Lemma}

\newtheorem*{recap}{Recap}
\theoremstyle{definition}
\newtheorem{defn}{Definition}

\newtheorem{assum}{Assumption}

\graphicspath{{figs/}}

\begin{document}

\title{Lyapunov-Based Stabilization and Control of Closed Quantum Systems}

\author{Elham~Jamalinia,  
             Peyman~Azodi,  
             Alireza~Khayatian, 
        and~Peyman~Setoodeh
\thanks{E. Jamalinia is with the Department of Electrical and Computer Engineering, Lehigh University, Bethlehem, PA, USA (e-mail: elj320@lehigh.edu).}
\thanks{P. Azodi is with the Department of Chemistry, Princeton University, Princeton, NJ, USA (e-mail: pazodi@princeton.edu).}
\thanks{A. Khayatian and P. Setoodeh are with the School of Electrical and Computer Engineering, Shiraz University, Shiraz, Iran (e-mail: khayatia@shirazu.ac.ir; psetoodeh@shirazu.ac.ir).}
}

\maketitle

\begin{abstract}
A Lyapunov-based method is presented for stabilizing and controlling of closed quantum systems. The proposed method is constructed upon a novel quantum Lyapunov function of the system state trajectory tracking error. A positive-definite operator in the Lyapunov function provides additional degrees of freedom for the designer. The stabilization process is analyzed regarding two distinct cases for this operator in terms of its vanishing or non-vanishing commutation with the Hamiltonian operator of the undriven quantum system. To cope with the global phase invariance of quantum states as a result of the quantum projective measurement postulate, equivalence classes of quantum states are defined and used in the proposed Lyapunov-based analysis and design. Results show significant improvement in both the set of stabilizable quantum systems and their invariant sets of state trajectories generated by designed control signals. The proposed method can potentially be applied for high-fidelity quantum control purposes in quantum computing frameworks. 
\end{abstract}


\section{Introduction}
\label{sec:introduction}

\par Studying dynamic systems at nano-scale using quantum physics has been one of the most intriguing problems in the last century due to their unique characteristics such as probabilistic description, inherent spin, and entanglement. Engineering applications of quantum systems have been spreading in a variety of areas \cite{dowling2003quantum}; including quantum computing \cite{divincenzo1995quantum}, quantum communications \cite{holevo1979capacity,pirandola2019end}, and quantum sensing \cite{giovannetti2006quantum,pirandola2018advances} as well as health care, medicine, and biology \cite{kislyakov2019physics,mcfadden2016life,poggiali2018optimal}. Alongside this increase of applications, urge for a solid control and stability analysis framework for quantum systems has increased in the last two decades leading to emergence of quantum control theory \cite{PhysRevLett.68.1500,Warren1581,belavkin1987non}. In this paper, Lyapunov theory is deployed to analyze and control closed-quantum systems, which refer to quantum systems with no uncontrolled interaction with the environment.

\par Lyapunov-based control of quantum systems was first proposed in~\cite{grivopoulos2003lyapunov}, where the Lyapunov function was constructed upon the expected value of a predetermined observable. Although this theory is applicable to a limited number of classes of quantum systems, it was an important step towards controlling quantum systems using the Lyapunov theory. Lyapunov-based control of quantum systems was further developed in \cite{mirrahimi2005lyapunov} and \cite{kuang2008lyapunov}. In these papers, different Lyapunov functions based on the tracking error, the expected value of an operator at the state of the system, and the quantum state fidelity, were proposed and their properties were studied. For pure-state quantum systems, Lyapunov-based stabilization techniques have been further analyzed for switching control techniques \cite{zhao2012switching} and optimal control \cite{hou2012optimal,wang2014optimal}. Moreover, for mixed-state quantum systems, Lyapunov-based methods have been deployed to design stabilizing controllers \cite{wang2010analysis}. Recently, Lyapunov-based controllers have been applied to superconducting qubits \cite{zeng2018quantum,ji2018lyapunov} and cavity quantum electrodynamic devices \cite{li2018driving,ran2018high}. Each of the proposed Lyapunov functions are suitable for a specific class of quantum systems, which limits their domain of applicability. Proposing Lyapunov functions capable of stabilizing a wider class of quantum systems is a challenging problem. Another challenging problem is to propose a Lyapunov function capable of stabilizing the system at a more general set of quantum states. In this paper, a novel Lyapunov function is proposed, which is constructed upon the expected value of a positive-definite operator $P$ of the tracking error. It is shown that this Lyapunov function is capable of stabilizing a wider range of quantum systems at a generic set of quantum states. Moreover, $P$ provides additional degrees of freedom that can be used as control variables to satisfy different performance measures in the controller design procedure.

The proposed theory in this paper studies controlling and stabilization of closed quantum systems regarding two distinct cases for the positive-definite operator $P$. In the first case, this operator commutes with the drift Hamiltonian (i.e., Hamiltonian of the closed system in the absence of any control). This case is similar to some of the existing theories in the literature \cite{mirrahimi2005lyapunov,kuang2008lyapunov,zhao2012switching}. However, compared to these methods, the proposed approach provides a stabilizing controller design scheme, which relaxes some of the restricting conditions on the quantum systems for which this method is applicable. It is worth noting that the proposed method does not modify the set of stable fixed points of the dynamics. In the second case, the operator $P$, does not commute with the drift Hamiltonian. In this case, the fixed points of the dynamics are restricted to the set of eigenstates of the drift Hamiltonian without requiring any restricting condition on the quantum system under study. The controller design scheme proposed in this case is applicable to a wide range of quantum systems, which is the main contribution of this paper. 

\par This paper is organized as follows. In section \ref{problem}, time evolution of closed quantum systems under Schr\"odinger equation is introduced and some definitions and assumptions on the systems under consideration are presented. In section \ref{Lyapunov}, a novel Lyapunov function is introduced and its properties are studied. Additionally, the invariant set of quantum states are defined in this section. In subsections \ref{case1} and \ref{case2}, two distinct cases for the positive-definite operator $P$ in the Lyapunov function are analyzed in details, and the control input and invariant set conditions are derived. In section \ref{results}, a five-state quantum system is studied as an example to show the capabilities of the proposed method in comparison with some of the existing methods in the literature. The paper concludes in the final section.

\section{Problem Formulation}\label{problem}

Our goal here is to manipulate the pure state $\ket{\psi} \in \mathbb{C}^N$ of an $N-$ dimensional closed quantum system. The time evolution of this class of systems is governed by the following linear first-order partial differential equation, the celebrated Schr\"{o}dinger equation:
\begin{equation}\label{Schrodinger}
  i\hbar \frac{\partial }{\partial t}\ket{\psi{(t)}}= H \ket{\psi{(t)}}, \; \ket{\psi{(0)}}=\ket{\psi_0},
 \end{equation}
 where $\hbar$ is the Planck constant, $H=H^{\dag}$ is the Hamiltonian of the system, and $\ket{\psi_0}$ is the initial state. The pure state of the system is normalized $\braket{\psi|\psi}=1$ and lives on the hypersphere $\mathbb{R}^{2N-1}$.  The state transition group generated by (\ref{Schrodinger}) is the class of exponential maps, known as Special Unitary group $SU(N)$, which preserves the norm of the system state.
 
Let us assume that the Hamiltonian of the system has a bilinear form as follows:
\begin{equation}\label{control}
   H=H_0+\sum_{k=1}^{r} {u_k H_k},
\end{equation}
where $H_0$ is the so called drift Hamiltonian, $u_k$'s are the control signals and $H_k$'s are the corresponding control Hamiltonians. The eigenvalues and their associated eigenvectors of the drift Hamiltonian are denoted by $\lambda_i$ and $\ket{\lambda_i}$,  respectively, hence, $H_0\ket{\lambda_i}=\lambda_i \ket{\lambda_i}$.

The following assumptions on the system are helpful for stability analysis in the upcoming sections.
\begin{assum} The quantum system given by (\ref{Schrodinger}) and (\ref{control}) is Equivalent State Controllable (ESC) \cite{d2007introduction,albertini2003notions}. ESC assures that for any arbitrary initial and final states $\ket{\psi_0}$ and $\ket{\psi_f}$, there exist the finite time $T$, real number $\epsilon$, and control signals $u_k$, such that the solution to (\ref{Schrodinger}) given (\ref{control}) obeys the following equation:
$$
\ket{\psi(T)}=e^{i\epsilon}\ket{\psi_f}.
$$
\end{assum}
\begin{assum}\label{lambda}
The drift Hamiltonian $H_0$ is \textit{$\lambda$-nondegenerate} and the energy levels are not equally spaced:
\begin{equation}{\label{nondeg}}
    {\lambda_i}-{\lambda_j}\neq{\lambda_m}-{\lambda_n}, \;\;  i-j\neq m-n.
\end{equation}
This implies that $H_0$ is also \textit{nondegenerate}:
\begin{equation}\label{lambda1}
    \lambda_m\neq \lambda_j, \;\; \forall m \neq n\leq N.
\end{equation}
\end{assum}

The following concepts will be used in what follows:
\begin{defn}[Equivalence classes of quantum states]
The equivalence class of $\ket{\psi}$, denoted by $[\ket{\psi}]$, is defined as the set of all normalized complex vectors that are different from $\ket{\psi}$ by at most a complex global phase:
$$[\ket{\psi}]\triangleq \{e^{i\varrho}\ket{\psi}\mid\varrho\in\mathbb{R}\}. $$
\end{defn}

One can easily deduce that two quantum states $\ket{\psi_1}$ and $\ket{\psi_2}$ are equivalent if and only if $\ket{\psi_2}\in[\ket{\psi_1}]$. 
This definition is physically intuitive since all members of an equivalence class, although mathematically different, are indistinguishable by quantum von Neumann measurement and they describe the same state of knowledge about the quantum system. This unique physical property of quantum systems calls for a control theory that is invariant with respect to global phases. In other words, control signals designed to stabilize a quantum system must ignore the global phase of the system initial state \cite{azodi2017lyapunov}. The next definition describes the  \textit{Equivalently Satisfied Statements} to address this issue. 
\begin{defn}[Equivalently Satisfied Statements]\label{equiv}
A mathematical statement is equivalently satisfied for $\ket{\psi}$ and $\ket{\psi_f}$, if for every element in $[\ket{\psi}]$, there exists a corresponding element in $[\ket{\psi_f}]$ that satisfies the statement.  
\end{defn}

Due to the intrinsic random nature of quantum measurement, in designing control signals for quantum systems, the goal is to achieve the desired post-measurement statistics. This fact dictates that designed signals should transfer the measurement statistics of the initial system $[\ket{\psi}]$ -rather than the mathematical representation of one of its physical realizations - to the desired final measurement statistics $[\ket{\psi_f}]$. Designed control signals will hold this important feature, if they equivalently satisfy the statements in the sense described in Definition \ref{equiv}.

\par The following assumptions pave the way for system stabilization.
\begin{assum}\label{H_0,f}
The final desired state $\ket{\psi_f}$ is an eigenstate of the drift Hamiltonian $H_0$ with the corresponding eigenvalue $\lambda_f$: 
\begin{equation}\label{H_0}
    H_0\ket{\psi_f}=\lambda_f\ket{\psi_f}.
\end{equation}
\end{assum}
\begin{assum}
$\ket{\psi_f}$ is not an eigenstate of the control Hamiltonian operators. 
\end{assum}
As will be discussed in the following section, this assumption guarantees that all control Hamiltonians take part in the stabilization process. If the condition is not satisfied, then, the corresponding Hamiltonians will not contribute to the stabilization process (i.e., their associated terms will vanish). 
\begin{assum}\label{assumH}
Control Hamiltonian operators do not have any common eigenvectors.
\end{assum}

\par In the following section, a Lyapunov-based stabilization method is proposed for the closed quantum systems described in this section. Detailed stability analysis is also provided for such systems.

\section{Lyapunov-Based Stabilization}\label{Lyapunov}

In this section, a novel Lyapunov function is proposed based on the tracking error $\ket{\psi{(t)}-\psi_f}$, which is characterized by a positive-definite operator $P$. Using Lyapunov theory, a general form for the control law $u_k$ will be obtained and the invariant set of the dynamic system will be found to assure the stability of the system. In the following subsections, it will be shown that how the structure of operator $P$ affects the invariant set by considering two distinct cases.

\subsection{Lyapunov Function}\label{Lyapunov function}

There are several proposed Lyapunov candidates for quantum systems in the literature. In this paper, a novel Lyapunov function is proposed as follows:
\begin{equation}\label{2}
  V(\ket{\psi(t)})=\bra{\psi{(t)}-\psi_f}P\ket{\psi{(t)}-\psi_f},
\end{equation}
where $P$ is a positive-definite operator. This quadratic Lyapunov function minimizes the error in the post-measurement statistics. Additional degrees of freedom provided by $P$, allow for satisfying different performance measures and ultimately lead to a high-performance controller design scheme for quantum systems. If $P=\frac{I}{2}$, then, the Lyapunov function will be the same as in \cite{kuang2008lyapunov}, and if $\ket{\psi_f}$ is not considered in the Lyapunov function, it will reduce to the form proposed in \cite{grivopoulos2003lyapunov}.

\begin{lemma}\label{Lya}
The Lyapunov function (\ref{2}) is equivalently zero in $[\ket{\psi_f}]$.
\end{lemma}
\begin{proof}
Based on Definition \ref{equiv}, if the Lyapunov function (\ref{2}) is equivalently zero, then for every element in $[\ket{\psi}]$, there must exist a corresponding element in $[\ket{\psi_f}]$ satisfying this condition. In other words, for every $\alpha \in \mathbb{R}$, there should be a $\beta \in \mathbb{R}$ such that
$$
 \bra{e^{-i\alpha}\psi-e^{-i\beta}\psi_f}P\ket{e^{i\alpha}\psi-e^{i\beta}\psi_f}=0,
$$ 
which is zero when $\ket{e^{i\alpha}\psi-e^{i\beta}\psi_f}=0$. Regarding the fact that for a positive-definite operator $P$, we have $x^{\star}P x=0\iff x=0$, the Lyapunov function (\ref{2}) is zero if and only if 
\\
$\ket{e^{i\alpha}\psi-e^{i\beta}\psi_f}=0\iff {e^{i\alpha}}\ket\psi={e^{i\beta}\ket{\psi_f}}\iff \ket\psi= {e^{i(\alpha-\beta)}}\ket{\psi_f}\iff \psi \in [\ket{\psi_f}]. $ 
\end{proof}

Without considering Definition \ref{equiv}, the Lyapunov function (\ref{2}) is zero if and only if $\ket{\psi}=\ket{\psi_f}$. However, in quantum systems, physically, there is no difference between members of an equivalence class. Thus, the Lyapunov function should be zero in all members of the equivalence class $[\ket{\psi_f}]$. As mentioned previously, this issue can be addressed by defining the equivalently satisfied statements as in Definition \ref{equiv}. 

\subsection{Time Derivative of the Lyapunov Function}\label{general}

In order to investigate the stability of the quantum system described in (\ref{Schrodinger}) and (\ref{control}), and find the fixed points of the dynamic system, time evolution of Lyapunov function needs to be analyzed:
\begin{eqnarray}
 V&=&\bra{\psi{(t)}-\psi_f}P\ket{\psi{(t)}-\psi_f}\\&=&\bra{\psi{(t)}}P\ket{\psi{(t)}}-\bra{\psi_f}P\ket{\psi{(t)}}\nonumber\\&-&\bra{\psi{(t)}}P{\ket{\psi_{f}}}+\bra{{\psi_f}}P\ket{\psi_f} \nonumber.
\end{eqnarray}
The time derivative of $V$ will be: 
\begin{eqnarray}
\dot{V}&=&\dot{\bra{\psi{(t)}}}P\ket{\psi{(t)}}+\bra{\psi{(t)}}P\dot{\ket{\psi{(t)}}} \\&-&\bra{\psi_f}P\dot{\ket{\psi{(t)}}}-\dot{\bra{\psi{(t)}}}P{\ket{\psi_{f}}}\nonumber.
\end{eqnarray}
Substituting for $\dot{\ket{\psi{(t)}}}$ from (\ref{Schrodinger}), we have:
\begin{eqnarray}
\dot{V}&=&\frac{i}{h} \Big( \bra{\psi{(t)}}HP\ket{\psi{(t)}}-\bra{\psi{(t)}}PH\ket{\psi{(t)}}\\&+& \bra{\psi_f}PH\ket{\psi{(t)}}-\bra{\psi{(t)}}HP{\ket{\psi_{f}}} \Big) 
\nonumber\\&=&\frac{i}{h} \Big( \bra{\psi{(t)}}(HP-PH)\ket{\psi{(t)}} \nonumber\\&-& \big(\bra{\psi{(t)}}HP{\ket{\psi_{f}}} - (\bra{\psi{(t)}}HP{\ket{\psi_{f}}})^*\big) \Big) \nonumber \\ &=&\frac{i}{h}\bra{\psi{(t)}}[H,P]\ket{\psi{(t)}}+\frac{2}{h} \text{Im}\{\bra{\psi{(t)}}HP{\ket{\psi_{f}}}\} \nonumber,
\end{eqnarray}
where $\text{Im}\{.\}$ denotes the imaginary part. Substituting $H$ form (\ref{control}) into the above equation and extending the commutations, we obtain: 
\begin{eqnarray}
\dot{V}&=&\frac{i}{h}\bra{\psi{(t)}}[H_{0},P]\ket{\psi{(t)}} 
\\&+& \frac{i}{h}\displaystyle\sum_{k=1}^{r} {\bra{\psi}[H_k,P]\ket{\psi}u_k} \nonumber
\\&+&\frac{2}{h}\text{Im}\{\bra{\psi{(t)}}H_0P{\ket{\psi_{f}}}\}\nonumber
\\&+&\frac{2}{h}\displaystyle\sum_{k=1}^{r}\text{Im}\{\bra{\psi{(t)}}H_{k}P{\ket{\psi_{f}}}\}u_k \nonumber
\end{eqnarray}
\begin{eqnarray}\label{vdot}
\dot{V}&=&\frac{i}{\hbar}\bra{\psi}[H_0,P]\ket{\psi}+\frac{2}{\hbar}\text{Im}\{\bra{\psi}H_0P\ket{\psi_f}\} 
\\&+& \displaystyle\sum_{k=1}^{r} \Big(\frac{i}{\hbar}\bra{\psi}[H_k,P]\ket{\psi} + \frac{2}{\hbar}\text{Im}\{\bra{\psi}H_k P\ket{\psi_f}\}\Big)u_k \nonumber
\end{eqnarray}

The goal is to design ${u_k}$ such that the system is stabilized at $[\ket{\psi_f}]$. This goal can be achieved by choosing the following form for the control law $(u_k)$:
\begin{equation}\label{uk}
    u_k(t)=-K_k f_k\Big(i\bra{\psi}[H_k,P]\ket{\psi}+2\text{Im}\{\bra{\psi}H_k P\ket{\psi_f}\}\Big),
\end{equation}
where $K_k$'s are positive real constants and $f_k$'s are odd functions. 
Terms $i\bra{\psi}[H_k,P]\ket{\psi}$ and $\text{Im}\{\bra{\psi}H_k\ket{\psi_f}\}$ were also used in the control laws proposed in \cite{grivopoulos2003lyapunov} and \cite{kuang2018lyapunov}, respectively. Proper choice of constants $K_k$'s can guarantee that the time derivative of the Lyapunov function will be negative.

In order to find the invariant set, equations $\dot{V}=0$ and $u_k=0,\forall{k}\in \{1,...,r\}$ must be solved. Regarding the form of $ u_k(t)$ in (\ref{uk}), $\forall{k}\in \{1,...,r\}$, we have:
\begin{equation}\label{u-invar}
    u_k=0 \iff i\bra{\psi}[H_k,P]\ket{\psi}+2\text{Im}\{\bra{\psi}H_k P\ket{\psi_f}\}=0.
\end{equation}
Using equations (\ref{u-invar}) and (\ref{vdot}), $\dot{V} |_{u_k=0}=0$ leads to:
\begin{equation}\label{v-invar}
\frac{i}{\hbar}\bra{\psi}[H_0,P]\ket{\psi}+\frac{2}{\hbar}\text{Im}\{\bra{\psi}H_0P\ket{\psi_f}\}=0.
\end{equation}
The intersection of solutions to equations (\ref{u-invar}) and (\ref{v-invar}) characterizes the invariant set of the dynamic system, which can be modified by proper design of operator $P$.

Two different general cases can be considered for operator $P$: 
\begin{enumerate}[label=(\roman*)]
  \item Operator $P$ and the drift Hamiltonian of the system $H_0$ share the same set of eigenvectors, $[H_0,P]=0$.
  \item Operator $P$ and the drift Hamiltonian of the system $H_0$ have different set of eigenvectors, $[H_0,P]\neq0$.
\end{enumerate}
The following subsections discuss these two cases in detail, where $[A,B]=AB-BA$.

\subsection{Case~I $[H_0,P]=0$}\label{case1}

In this case, operator $P$ is selected such that $[H_0,P]=0$, which means that $H_0$ and $P$ can share the same set of eigenvectors. It is shown that for the Lyapunov function (\ref{2}), the invariant set of the system will be restricted.    
Based on Assumption \ref{H_0,f}, the desired final state $\ket{\psi_f}$ is an eigenvector of $H_0$ and due to the common eigenbasis of $H_0$ and $P$, the desired final state is an eigenvector of $P$ as well:
\begin{equation}\label{P1}
    P\ket{\psi_f}=p_f\ket{\psi_f},
\end{equation}
where $p_f > 0$ is the corresponding eigenvalue, regarding that $P$ is positive-definite.

Based on the general form of the time derivative of the Lyapunov function in (\ref{vdot}) and  equation (\ref{P1}), $\dot{V}$ can be written as:
\begin{eqnarray}\label{vdot1}
\dot{V}&=&\frac{2}{\hbar} \lambda_f{p_f}\text{Im}\{\braket{\psi|\psi_f}\} \\ &+& \displaystyle\sum_{k=1}^{r} {u_k(\frac{i}{\hbar}\bra{\psi}[H_k,P]\ket{\psi}+\frac{2}{\hbar} p_f\text{Im}\{\bra{\psi}H_k\ket{\psi_f}})\}. \nonumber
\end{eqnarray}

In this case, the control law is obtained from the general control law in (\ref{uk}) as follows:
\begin{equation}\label{uk1}
    u_k(t)=-K_k f_k(i\bra{\psi}[H_k,P]\ket{\psi}+2p_f\text{Im}\{\bra{\psi}H_k\ket{\psi_f}\}).
\end{equation}
Based on the general form of the invariant set obtained from equations (\ref{u-invar}) and (\ref{v-invar}), the invariant set in this case is characterized as:
\begin{eqnarray}\label{u2}
   \Big\{\ket{\psi} | i\bra{\psi}[H_k,P]\ket{\psi}+2p_f\text{Im}\{\bra{\psi}H_k\ket{\psi_f}\}=0; \\ \forall{k}\in \{1,...,r\}\Big\} \nonumber
\end{eqnarray}
Regarding equations (\ref{v-invar}) and (\ref{P1}), another necessary condition for $\ket{\psi}$ to be an invariant point is (note that $\lambda_f{p_f}>0$):
\begin{equation}
    \text{Im}\{\braket{\psi|\psi_f}\}=0.
\end{equation}
As discussed previously, due to the global phase invariance property, the invariant set should be equivalently satisfied. By substituting $\ket{\psi}$ with its equivalent states $e^{i\varrho}\ket{\psi}$ in $\bra{\psi}[H_k,P]\ket{\psi}$, we have:
\begin{eqnarray}
e^{-i\varrho}\bra{\psi}[H_k,P]\ket{\psi}e^{i\varrho} \nonumber &=& e^{-i\varrho}e^{+i\varrho}\bra{\psi}[H_k,P]\ket{\psi} \\ &=& \bra{\psi}[H_k,P]\ket{\psi}. 
\end{eqnarray}
Therefore, in equation (\ref{u2}), the first term is invariant in an equivalence class. Thus, in order for this equation to equivalently remain zero, the term $\text{Im}\{\bra{\psi}H_k P\ket{\psi_f}\}$ should be zero in the equivalence class.

To summarize, the invariant set of the system is the intersection of the following sets:
\begin{equation}\label{invar1}
   \{\ket{\psi}| \; \bra{\psi}[H_k,P]\ket{\psi}=0; \forall{k}\in \{1,...,r\}\},
\end{equation}
\begin{equation}\label{invar2}
    \{\ket{\psi}| \; \text{Im}\{\bra{\psi}H_k\ket{\psi_f}\}=0 ; \forall{k}\in \{1,...,r\}\},
\end{equation}
\begin{equation}\label{invar3}
    \{\ket{\psi}| \; \text{Im}\{\braket{\psi|\psi_f}\}=0\}.
\end{equation}
Set (\ref{invar1}) was also derived in \cite{grivopoulos2003lyapunov}. As discussed before, equation (\ref{invar1}) is equivalently invariant. Based on Definition \ref{equiv}, in order for equations (\ref{invar2}) and (\ref{invar3}) to be equivalently invariant, for all $\alpha \in[0,2\pi)$, there should exist $\beta\in[0,2\pi)$ such that:
\begin{equation}\label{alpha}
    \text{Im}\Big\{\bra{\psi}H_k \ket{\psi_f}e^{i(\beta-\alpha)}\Big\}=0 ; \; \forall{k}\in \{1,...,r\},
\end{equation}
\begin{equation}\label{beta}
     \text{Im}\Big\{\braket{\psi|\psi_f}e^{i(\beta-\alpha)}\Big\}=0 ; \; \forall{k}\in \{1,...,r\}.
\end{equation}
Equations (\ref{alpha}) and (\ref{beta}) cannot simultaneously hold except for one of the following cases:
 \begin{enumerate}[label=(\roman*)]
 \item If $\braket{\psi|\psi_f}=0$, and
 \begin{equation}\label{cond1}
     \angle{\bra{\psi}H_k \ket{\psi_f}}=\angle{\bra{\psi}H_j \ket{\psi_f}}; \; \forall{k,j}\in \{1,...,r\},
 \end{equation}
 \item If $\braket{\psi|\psi_f}\neq0$ and
  \begin{equation}\label{cond2}
     \angle{\bra{\psi}H_k \ket{\psi_f}}=\angle{\braket{\psi|\psi_f}}; \; \forall{k}\in \{1,...,r\}.
 \end{equation}
 \end{enumerate}

So far, the invariant set of the system has been characterized for the case $[H_0,P]=0$. In what follows, the invariant set is investigated from another point of view. Let us expand $\ket{\psi_0}$ in the basis of $\ket{\lambda_m}$ \cite{grivopoulos2003lyapunov}:
 \begin{equation}\label{form0}
     \ket{\psi_0}=\displaystyle\sum_{m=1}^{n} c_m\ket{\lambda_m},
 \end{equation}
where $c_m=\braket{\lambda_m,\psi_0}\in \mathbb{C}$ and $\ket{\lambda_m}$ was defined regarding equation (\ref{control}).
In the invariant set ($u_k=0$), based on the Schr\"{o}dinger equation (\ref{Schrodinger}) and equation (\ref{form0}), the time evolution of $\ket{\psi}$ is obtained from:
\begin{equation}\label{form}
     \ket{\psi}=\displaystyle\sum_{m=1}^{n} c_me^{-i\frac{\lambda_m}{\hbar}t}\ket{\lambda_m}.
 \end{equation}
Substituting for $\ket{\psi}$ from (\ref{form}) into equation (\ref{invar1}), we have:
 \begin{multline}\label{c1}
     \bra{\psi}[H_k,P]\ket{\psi} = \\ \displaystyle\sum_{m,j}c_m^{\star}c_j(p_j-p_m)e^{i(\lambda_m-\lambda_j)t}\bra{\lambda_m}H_k\ket{\lambda_j}=0, 
 \end{multline}
which should remain zero $\forall t>0$. As a result, based on Assumption \ref{lambda}, the following condition should be necessarily satisfied in the invariant set:
 \begin{equation}\label{h}
      c_m^{\star}c_j\bra{\lambda_m}H_k\ket{\lambda_j}=0;  \;  \forall m\neq j ,\forall{k}.
 \end{equation}
This condition highly restricts the class of quantum systems under study, because it requires transition amplitude for all control Hamiltonians to be between eigenstates of the drift Hamiltonian. Although (\ref{h}) is a very restricting condition, it is required by other similar methods reported in the literature \cite{kuang2008lyapunov}. In the following subsection, it will be shown that, how this condition can be relaxed by suitably designing $P$ such that $[H_0,P]\neq0$.

\begin{recap}
Using the proposed Lyapunov function (\ref{2}), the invariant set of the system not only is characterized by the conditions proposed in \cite{kuang2008lyapunov} and \cite{shuang2007quantum} (without the constraint ${\bra{\lambda_m}}H_k\ket{\lambda_j}\neq0$), but also has to satisfy conditions (\ref{cond1}) and (\ref{cond2}). As a result, using the control law (\ref{uk1}), the invariant set of the system can be significantly modified and restricted.
\end{recap}

\subsection{Case~II $[H_0,P]\neq0$}\label{case2}

The second case for operator $P$ is discussed in this subsection. Here, this operator is designed such that $P$ and drift Hamiltonian $H_0$ do not share any set of common eigenbasis. At first glance, this case might seem mathematically more challenging, but it is shown that by using such an operator $P$ in the Lyapunov function, the invariant set of the system can be restricted to the desired final state $\ket{\psi_f}$.

In this case, the time derivative of the Lyapunov function $\dot{V}$ and the control law $u_k$ are defined as in (\ref{vdot}) and (\ref{uk}), respectively. The invariant set of the system will be analyzed in the sequel. Following the strategy presented in subsection \ref{general}, the basic invariant set conditions in this case are obtained as (\ref{u-invar}) and (\ref{v-invar}). By analyzing these equations in the equivalence class $[\ket{\psi}]$, as discussed in subsection \ref{general}, the two invariant set equations in (\ref{u-invar}) and (\ref{v-invar}) can be presented as the following four equations:
 \begin{equation}\label{invarr0}
      \bra{\psi}[H_0,P]\ket{\psi}=0,
 \end{equation}
 \begin{equation}\label{invarr1}
    \bra{\psi}[H_k,P]\ket{\psi}=0; \; \forall{k}\in \{1,...,r\},
 \end{equation}
 \begin{equation}\label{invarr2}
      \text{Im}\{\bra{\psi}H_0 P\ket{\psi_f}\}=0,
 \end{equation}
 \begin{equation}\label{invarr3}
     \text{Im}\{\bra{\psi}H_k P\ket{\psi_f}\}=0; \; \forall{k}\in \{1,...,r\}.
 \end{equation}
Equation (\ref{invarr0}) is analyzed by expanding $\ket{\psi}$ based on the set of eigenvectors of the drift Hamiltonian $\ket{\lambda_m}$. Note that this expansion is not based on the eigenvectors of $P$, because operator $P$ and drift Hamiltonian of the system do not share any set of eigenstates.
Substituting $\ket{\psi}$ from (\ref{form}) into equation (\ref{invarr0}), we obtain:
\begin{multline}\label{eq}
     \bra{\psi}[H_0,P]\ket{\psi}= \\ \displaystyle\sum_{m,j}c_m^{\star}c_j(\lambda_m-\lambda_j)e^{i(\lambda_m-\lambda_j)t}\bra{\lambda_m}P\ket{\lambda_j}=0.
\end{multline}
If $\ket{\psi}$ is in the invariant set, this condition should hold $\forall t>0$.

The following theorem provides a sufficient condition on $P$, in order to restrict the invariant set of the system to the set of eigenvectors of $H_0$.

\begin{theorem}
If operator $P$ satisfies the following condition,
\begin{equation}\label{P}
     \bra{\lambda_m}P\ket{\lambda_j}\neq 0; \; \forall m\neq j \in \{1,...,n\},
\end{equation}
then, the invariant set of the dynamics given by (\ref{Schrodinger}) is uniquely restricted to the set of eigenvectors of $H_0$. However, linear combinations of the eigenvectors of $H_0$ will not be in the invariant set.
\end{theorem}
\begin{proof}
When the system state trajectory reaches the invariant set, equation (\ref{eq}) should hold. Given the fact that equation (\ref{eq}) must hold for all $\ket{\psi}$ in the invariant set of the system $\forall{t>T}$, if operator $P$ is designed such that equation (\ref{P}) is satisfied, then, based on equation (\ref{lambda1}), the following condition must hold in order to ensure condition (\ref{eq}):  
\begin{equation}\label{c}
 c_m^{\star}c_j=0; \; \forall m,j \in \{1,...,n\}.
 \end{equation}
This condition is satisfied if $c_m\neq0$ for exactly one $m\in \{1,...,n\}$. Using the normalization condition $\displaystyle\sum_{i=1}^{n} c_m^2=1$, it is concluded that $c_m=1$ and $\ket{\psi}=[\ket{\lambda_m}]$, where $\ket{\lambda_m}$ is the corresponding eigenvector of drift Hamiltonian. 
 \end{proof}

Equations (\ref{h}) and (\ref{P}) are very similar except that in equation (\ref{h}), the constraint is on the control Hamiltonian operators $H_k$, which are usually associated with the nature of the system and  cannot be freely chosen. However, in equation (\ref{P}) the constraint is on operator $P$, which does not interfere with the nature of the system. 
 As shown before, conditions (\ref{invarr0}) and (\ref{invarr3}) characterize the invariant set of the system. So far, using equation (\ref{invarr0}), it has been found that by properly choosing $P$, the invariant set will uniquely include eigenvectors of $H_0$. Therefore, it is sufficient to check the proposed conditions for eigenvectors of $H_0$.
 
 In the next section, designing the operator $P$ to satisfy equation (\ref{P}) is demonstrated via an example.

\section{Simulation Results}\label{results}
 
  In this section, first a five-state quantum system from \cite{tersigni1990using,kuang2008lyapunov} is introduced, and then, the Lyapunov-based stabilization methods proposed in \cite{grivopoulos2003lyapunov,mirrahimi2005lyapunov,shuang2007quantum} are compared to the proposed method in this paper.

Consider the time evolution of a five-level quantum system with one control input ($u_1$) as follows:  
\begin{equation}\label{system}
  i \frac{\partial }{\partial t}\ket{\psi{(t)}}= (H_0+u_1H_1)\ket{\psi{(t)}},
 \end{equation}
where $\hbar$ is normalized to $1$ for simplicity. Drift and control Hamiltonians of the system are respectively defined as:
  \begin{equation}\label{systemH0}
      H_0=\begin{pmatrix}
      1 & 0 & 0 & 0 & 0 \\
      0 & 1.2 & 0 & 0 & 0 \\ 
       0 & 0 & 1.3 & 0 & 0 \\
       0 & 0 & 0 & 2 & 0 \\
       0 & 0 & 0 & 0 & 2.15 \\
      \end{pmatrix},
  \end{equation}
  \begin{equation}\label{systemH1}
      H_1=\begin{pmatrix}
      0 & 0 & 0 & 1 & 1 \\
      0 & 0 & 0 & 1 & 1 \\ 
       0 & 0 & 0 & 1 & 1 \\
       1 & 1 & 1 & 0 & 0 \\
       1 & 1 & 1 & 0 & 0 \\
      \end{pmatrix}.
  \end{equation}
  State of the system is denoted as $\ket{\psi_{(t)}}=(c_1,c_2,c_3,c_4,c_5)^T$. It is assumed that the system is initially in the state $\ket{\psi_0}=(1,0,0,0,0)^T$. The goal is to lead the system from this initial state to the final state $\ket{\psi_f}=(0,0,0,0,1)^T$. In what follows, we show how the designed control signal using the proposed theory stabilizes this system. In this section, operator $P$ is selected such that $[H_0,P]\neq0$, and additionally, condition ({\ref{P}}) is satisfied. Therefore, based on Theorem 2, this method stabilizes the system. Orthonormal eigenvectors of operator $P$ are chosen as follow:
\begin{eqnarray}
     \ket{P_1}&=& \begin{bmatrix}
     1\\
     1\\
     -1\\
     -1\\
     0
\end{bmatrix}, \;
\ket{P_2}= \begin{bmatrix}
     1\\
     -1\\
     1\\
     -1\\
     0
\end{bmatrix}, \;
\ket{P_3}= \begin{bmatrix}
     1\\
     1\\
     1\\
     1\\
     0
\end{bmatrix}, \\
\ket{P_4}&=& \begin{bmatrix}
     -1\\
     1\\
     1\\
     -1\\
     -1
\end{bmatrix}, \;
\ket{P_5}= \begin{bmatrix}
     -1\\
     1\\
     1\\
     -1\\
    -4
\end{bmatrix}, \nonumber
\end{eqnarray}
with the corresponding eigenvalues:
  \begin{equation}
      p_1=3, p_2=1, p_3=1, p_4=1, p_5=0.05.
  \end{equation}
These values were chosen by trial and error to get the best performance for the system. In matrix form, operator $P$ can be written as:
   \begin{equation}
      P=\begin{pmatrix}
      5.2 & 0.8 & 2.8 & -0.8 & 0.3 \\
      0.8 & 5.2 & -0.8 & 2.8 & -0.3 \\ 
       2.8 & -0.8 & 5.2 & 0.8 & -0.3 \\
       -0.8 & 2.8 & 0.8 & 5.2 & 0.3 \\
       0.3 & -0.3 & -0.3 & 0.3 & 1.7 \\
      
      \end{pmatrix}.
  \end{equation}
Figures \ref{fig:el1} and \ref{fig:el2} illustrate the state trajectories and time evolution of the Lyapunov function (\ref{2}), respectively. 
  \begin{figure}
      \centering
      \includegraphics[width=1\linewidth]{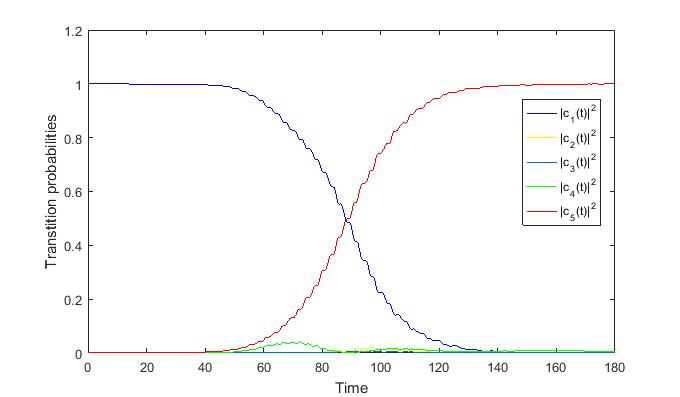}
      \caption{State trajectories of the system, subject to control law (\ref{uk}). The illustration shows complete transition from the initial state $\ket{\psi_0}=(1,0,0,0,0)^T$ to the desired state $\ket{\psi_f}=(0,0,0,0,1)^T$.}
      \label{fig:el1}
  \end{figure}
  \begin{figure}
      \centering
      \includegraphics[width=0.9\linewidth]{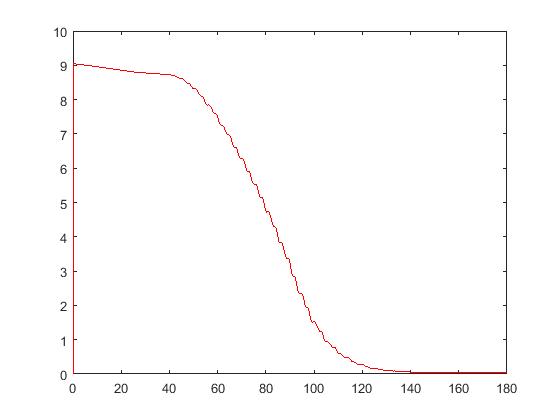}
      \caption{Time evolution of the Lyapunov function (\ref{2}). This figure illustrates how the value of the Lyapunov function decreases in time, which leads the system to the desired final state.}
      \label{fig:el2}
  \end{figure}

Other Lyapunov-based methods proposed in the literature \cite{grivopoulos2003lyapunov,kuang2008lyapunov,mirrahimi2005lyapunov}  are also applied to system (\ref{system})-(\ref{systemH1}), and simulation results are illustrated. These methods, although applicable to a wide variety of other quantum systems, do not stabilize the quantum system under consideration:
\begin{itemize}
\item The method proposed in \cite{grivopoulos2003lyapunov} is not applicable for this system, because the necessary condition $\bra{\lambda_m}H_1\ket{\lambda_j}\neq 0; \; \forall m\neq j$ is not satisfied for this quantum system. 
\item The stabilization scheme based on Lyapunov function ($8$) in \cite{kuang2008lyapunov} cannot be used for this five-state quantum system, because the control input will stay equal to zero and the Lyapunov function remains constant.
\item The stabilization scheme based on the Lyapunov function proposed in \cite{kuang2008lyapunov} results in an oscillating control signal, and the system will become unstable as shown in Figure \ref{fig:333}. 
\item Similarly, the stabilization scheme proposed in \cite{mirrahimi2005lyapunov} results in an oscillating control signal and the system will become unstable as shown in Figure \ref{fig:444}. 
\end{itemize}
Figures \ref{fig:el1}, \ref{fig:333}, and \ref{fig:444} confirm the efficiency and superiority of the proposed method.
 
\begin{figure}
     \centering
     \includegraphics[width=1\linewidth]{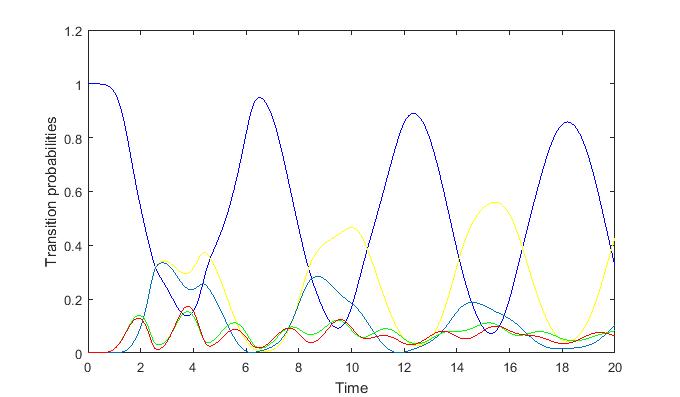}     \caption{State trajectory of the system, using the method proposed in \cite{kuang2008lyapunov}. As illustrated in the figure, the control scheme makes the system unstable.}
     \label{fig:333}
\end{figure}

\begin{figure}
     \centering
     \includegraphics[width=1\linewidth]{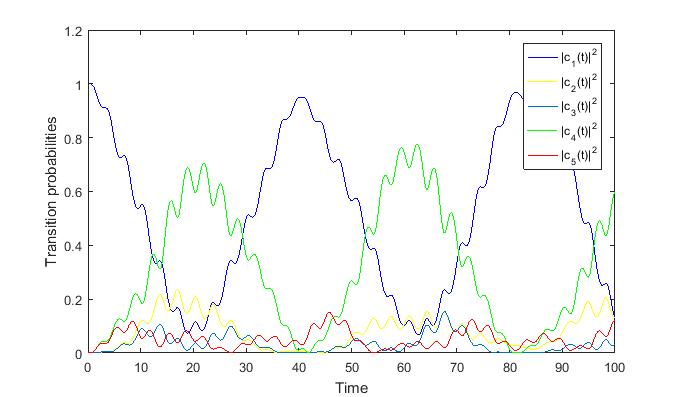}
     \caption{State trajectory of the system, using the method proposed in \cite{mirrahimi2005lyapunov}. As illustrated in the figure, this control scheme makes the system unstable.}
     \label{fig:444}
\end{figure}{}
    
\section{Conclusion} 
\par 
A Lyapunov-based stabilization scheme was proposed, which is applicable to a large class of closed quantum systems with minimal restrictions on the nature of the system. This scheme uses a positive-definite operator, which provides extra degrees of freedom for the designer. Depending on whether or not this operator and drift Hamiltonian of the system have common eigenvectors, two distinct cases can be considered. In the latter case, where the positive-definite operator in the Lyapunov function and drift Hamiltonian of the quantum system do not commute, it is shown that the set of invariant states of the system is limited to the set of eigenstates of the drift Hamiltonian. Additionally, the resulting quantum controller provides invariance with respect to the global phase of quantum states by virtue of equivalence classes of quantum states and equivalently satisfied statement criterion. Simulation results showed how the proposed method was capable of stabilizing a five-level quantum system, which cannot be successfully stabilized by other existing methods in the literature.
\par
For future research, a time-varying version of the positive-definite operator will be considered that can be adaptively tuned. Potential applications of the proposed method in quantum computing, especially  superconducting quantum computing, will also be explored.
  
\bibliographystyle{IEEEtran}
\bibliography{IEEEabrv,Main}

\end{document}